\pdfoutput=1
\documentclass[11pt,a4paper]{article}
\usepackage[a4paper,margin=2.7cm]{geometry}
\usepackage{lmodern}
\usepackage{authblk}
\usepackage[T1]{fontenc}
\usepackage[utf8]{inputenc}
\usepackage{amsmath,amssymb,amsthm}
\usepackage{graphicx}
\usepackage{booktabs}
\usepackage{algorithm}
\usepackage[noend]{algpseudocode}
\usepackage{array}
\usepackage{url}
\usepackage{placeins}
\usepackage{xcolor}
\usepackage{tikz}
\usetikzlibrary{arrows.meta,positioning,calc,fit}

\newtheorem{theorem}{Theorem}[section]
\newtheorem{proposition}[theorem]{Proposition}
\newtheorem{lemma}[theorem]{Lemma}
\theoremstyle{definition}
\newtheorem{definition}[theorem]{Definition}
\newtheorem{remark}[theorem]{Remark}
\newcommand{\Sig}{\Sigma}
\newcommand{\Sigfour}{\Sig_{4}}
\newcommand{\Sigeight}{\Sig_{8}}
\newcommand{\Sigtwelve}{\Sig_{12}}
\newcommand{\lang}{\mathcal{L}}
\newcommand{\anti}{\mathcal{A}}
\newcommand{\parikh}[1]{\boldsymbol{c}(#1)}


\usepackage[hidelinks]{hyperref}
\begin{document}

\title{Counting and Sampling Manufacturable Stacking Sequences}
\author[1]{Alessio Langiu\thanks{Corresponding author: \href{mailto:alessio.langiu@cnr.it}{alessio.langiu@cnr.it}}}
\author[2]{Felice Rubino}
\author[3]{Giacomo Canale}
\affil[1]{National Research Council of Italy - Institute of Marine Sciences (CNR-ISMAR), Palermo, Italy}
\affil[2]{Department of Industrial Engineering, University of Salerno, Fisciano, Italy}
\affil[3]{College of Science and Engineering, University of Derby, Derby, UK}
\date{}
\maketitle

\begin{abstract}
Designing the stacking sequence of a composite laminate is a discrete optimisation problem: each ply
takes one of a few orientations, each evaluation may cost a finite-element solve, and the
design must obey manufacturing rules and design guidelines, which optimisers enforce by penalty,
repair or rejection. We describe the admissible set exactly. The two manufacturing rules forbid a finite set of factors, so the sequences they
admit form a regular language, recognised by an automaton that tracks the last orientation and the
length of its run; balance and the
ten per cent rule constrain the Parikh vector and are carried by a dynamic program over automaton
state and counters. One table, polynomial in $N$, counts the
admissible sequences without producing them, draws them uniformly, and tells
whether a partial sequence can still be completed. We give pseudocode, exact counts for three
orientation sets, closed-form ceilings on the
composition, and four of the five single-rule entries a published table left uncomputed. At $48$ plies, $99.8\%$ of the nodes
of a layer-wise search tree checking the rules locally have no admissible completion on eight
orientations; the completion test removes them. A one-point crossover that cuts two admissible parents only where they share an automaton state and their child closes the counters returns admissible children by construction.
\end{abstract}



\section{Introduction}
\label{sec:intro}

A composite laminate is a stack of $N$ plies, each a sheet of parallel fibres laid at one of a
finite set of orientations, and its \emph{stacking sequence} lists those orientations from one
surface to the other. We borrow the vocabulary of formal languages: the allowed orientations form an
\emph{alphabet} $\Sig$, a stacking sequence is a \emph{word} over it, one letter per ply, and a set of
stacking sequences is a \emph{language}. The orientations determine the in-plane and bending stiffness of the laminate, hence its buckling
load, natural frequencies and deflection under load, and the design problem is to maximise one of
these responses, or to minimise $N$ subject to targets on them. It is a discrete
optimisation problem over $\Sig^{N}$, and each
evaluation of the objective may cost a finite-element solve. It has been searched with genetic
algorithms since the early 1990s \cite{callahan1992,leriche1993}, with permutation encodings
\cite{liu2000permutati,jing2015stacking}, particle swarms \cite{chang2009stacking}, scatter search
\cite{ramamohanrao2005scatter} and ant colonies \cite{aymerich2008optimizat}, and with branch and
bound: layer by layer from the outer surface inwards \cite{kim2005gust}, in a modified form that
starts from ``ideal'' layers and keeps one branch per level \cite[Sec.~2]{canale2018layup}, or in
the space of lamination parameters, where the admissible laminates form a fractal set
\cite{todoroki2004fractal}, non-symmetric laminates included \cite{matsuzaki2007unsym}. It has been
extended to blending across panels \cite{soremekun2002stacking,yang2016optimal}, to lamination
parameters as a continuous relaxation \cite{liu2019twolevel,albazzan2019efficient}, to
variable-stiffness laminates \cite{abdalla2007variable,rashed2022variable}, to strength as well as
buckling \cite{aymerich2008optimizat,an2014laminate}, to multiple objectives
\cite{abouhamze2007multiobje}, and to structures other than flat panels: cylindrical panels
\cite{abouhamze2007multiobje}, laminated grid plates \cite{ehsani2016stacking}, a wind-turbine blade
\cite{albanesi2018simultane} and an aero-engine fan-blade-like structure \cite{canale2018realistic}.

Not every word is a laminate that can be built or certified. Two manufacturing rules are local:
consecutive plies must not differ by more than $45^{\circ}$ of orientation, and no more than three
consecutive plies may share one. Designers add guidelines that are not local, namely mirror
symmetry, balance between $+\theta$ and $-\theta$, and a minimum share of ten per cent for each
principal direction; symmetry and balance are dropped where a directional stiffness or a
coupling is wanted, as in aeroelastic tailoring \cite{shirk1986aeroelastic}. Each of these methods spends part of its effort on candidates outside the set these rules
define: a genetic algorithm penalises the rules it does not build into its encoding
\cite{leriche1993} or repairs them \cite{todoroki1998repair}, a sampler builds candidates and
rejects the non-compliant, a branch and bound checks each partial sequence against the rules and
learns only at the leaves that a branch had no admissible completion. The rules enter as a
filter: as design rules to respect
\cite{irisarri2014,kim2007developme} or to relax \cite{verchery2011rules}, as constraints on the
fibre paths of tow-steered plies \cite[Sec.~5]{brooks2018manufactu}, or as a coupled ply-order and
ply-drop problem \cite{albanesi2018simultane}. We describe the admissible set exactly and derive from the description algorithms that produce
admissible sequences, and only those.

The set treated here is that of a single panel: blending between adjacent panels, ply drops
\cite{irisarri2014} and damage tolerance are not part of it (Section~\ref{sec:conclusions}).

Throughout the paper the reference size is $N=48$ plies. It is the ply count of the classical
benchmark of Le Riche and Haftka \cite{leriche1993}, and it lies far beyond the ply counts at which we enumerate exhaustively
(Section~\ref{sec:verification}). Nothing in the construction depends on it: the counts are given
from $24$ to $54$ plies, and every other quantity at $48$.

Zein et al.\ \cite{zein2014}, whose
constraint-programming method searches a backtracking enumeration tree pruned as soon as a prefix
breaks a rule, with rules of the same kind (no $90^{\circ}$ jump, at most four identical
consecutive plies, prescribed ply counts per orientation), write of it that ``it does not give
compute the number of admissible stacking sequences'' [\emph{sic}] \cite[p.~57]{zein2014}, and
frame the cost in terms of enumerating $4^{16}$ or $4^{32}$ candidates \cite[Sec.~4.2]{zein2014}. Irisarri
\cite{irisarri2009} tabulates the number of realisable stacking sequences rule by rule at four,
eight and sixteen plies, and leaves several cells marked as not computed. Fedon et al.\
\cite{fedon2021beam} note that the number of admissible laminates grows exponentially with the ply count, which makes
complete enumeration impossible; that is true of enumeration and does not follow for
counting. A granted
family assigned to Boeing builds a library of admissible sublaminates by growing an integer tree
and pruning the branches that violate the rules \cite{us10380319}, the rules being stated there as
ranges of ply counts per orientation, a cap on consecutive plies of one orientation, and a list of
sequences that are disallowed, a finite antidictionary in all but the name; a companion patent
bounds the ply counts of a laminate per orientation to decide feasibility before any optimisation
is run \cite{us10322792}. That tree is the unfolding of the automaton used here, and it is stored:
one node for every admissible prefix. Its specification puts the size of the \emph{unconstrained} space at $4^{10}$ sequences for a
ten-layer sublaminate over four orientations and does not report how many of them survive the
rules. Recent work maps the same design space onto a quantum state, or onto a tensor network, and carries
the manufacturing rules as penalty terms in an objective \cite{wulff2024}.

\begin{figure*}[!t]
  \centering
  \includegraphics[width=0.85\textwidth]{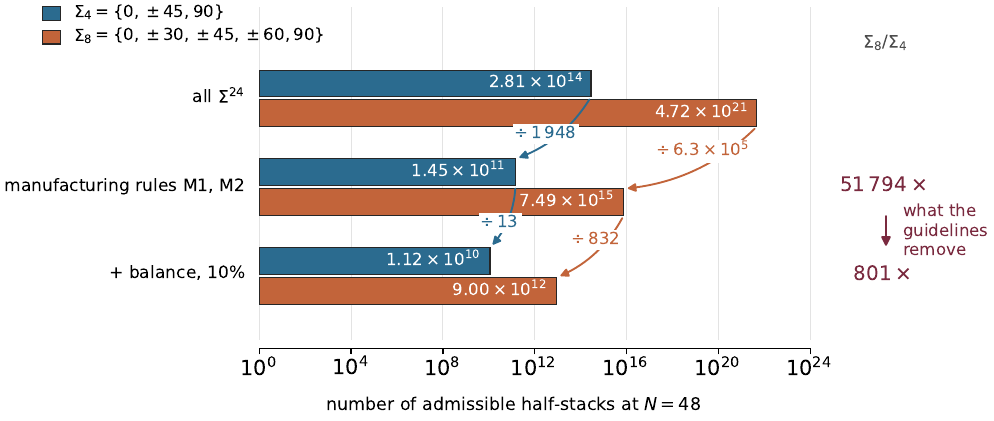}
  \caption{The admissible set inside all stacking sequences, at $N=48$. Each bar counts the
    half-stacks of $24$ plies that survive one restriction more than the bar above it. In the
    right-hand column, the extended set is larger by a factor of $51\,794$ on the
    manufacturing language and of $801$ once the guidelines are imposed. Exact counts, as in Table~\ref{tab:counts}.}
  \label{fig:nesting}
\end{figure*}

\subsection{Our results}

We show that the two manufacturing rules forbid a finite set of factors, so that the sequences they
admit form a regular language, recognised by an automaton whose state is the last orientation together with
the length of the run it terminates (Section~\ref{sec:language}), while balance and the $10\%$ rule constrain the Parikh vector and are carried exactly by a dynamic program over automaton
state and counters. Symmetry reduces a laminate to its half-stack, and the mirror leaves only a condition on the last
run. The
construction therefore covers symmetric and non-symmetric, balanced and unbalanced laminates, over
any finite set of orientations. We treat sets of four, eight and twelve orientations.

\emph{Algorithms.} One dynamic program, built once per alphabet and ply count, supports three
operations (Section~\ref{sec:counting}, Algorithms~\ref{alg:dfa}--\ref{alg:sample}): it returns
the number of admissible sequences without producing any; it draws an admissible sequence uniformly
at random in $O(N)$, with no rejection; and it decides whether a partial sequence can still be
completed, by a lookup after a walk of its length. A fourth operation, recombination of two
admissible parents, returns an admissible child on every call, a fresh exact draw when no admissible cut
exists (Algorithm~\ref{alg:cut}).

\emph{Counts.} We give exact counts for the three alphabets, for symmetric laminates at ply counts from $24$ to $54$, and for non-symmetric ones, balanced and unbalanced,
verified against exhaustive enumeration (Section~\ref{sec:verification}). Under the rules of
\cite{irisarri2009} our methods reproduce five of the six rule columns of its table at four
and eight plies, and fill four of the five single-rule entries it leaves uncomputed at sixteen
(Section~\ref{sec:completing}). At $48$ plies
the guidelines divide the symmetric language by $13$ and by $832$ on the four- and eight-orientation
sets, so the size advantage of the larger set falls from $51\,794$ to $801$
(Figure~\ref{fig:nesting}). We also give closed-form ceilings on the composition.

\emph{Tree search.} At $48$ plies, $90.8\%$ of the nodes of a branch and bound that assigns plies from the surface
inwards and checks the rules at each node have, on four orientations,
no admissible completion, $99.8\%$ on eight and more than $99.99\%$ on twelve
(Section~\ref{sec:tree}). The completion test of the program removes them, and since two prefixes
in the same cell have the same completions, the tree of completable prefixes collapses onto the
cells: $1.3\times 10^{6}$ of them against $2.9\times 10^{13}$ prefixes on eight orientations.

\emph{Recombination.} Two identical parents are in the same automaton state at every cut, so a population that has converged recombines only into copies of its members, and the diversity of a search inside the language has to come from outside the operator (Section~\ref{sec:search}). The exact sampler cannot fail and reports an empty language as empty (Section~\ref{sec:sampling}).

The techniques are classical and we claim none of them. Counting the words that avoid a finite set
of factors is done by a transfer matrix \cite[Ex.~V.11]{flajolet2009analytic} or by string
correlations \cite{guibas1981}; the forbidden-factor description we use, and the automaton that
recognises the words avoiding a finite antidictionary, are those of \cite[Sec.~2]{crochemore1998};
and transfer matrices have been used to enumerate and generate constrained configurations in other
fields \cite{kloczkowski1998transfer}. The sampler of Section~\ref{sec:sampling} belongs to the line
that generates combinatorial objects uniformly by counting them first
\cite{wilf1977unified,flajolet1994,hickey1983uniform,jerrum1986generatio,oudinet2013dichotomic}, to
which Boltzmann samplers are the standard alternative \cite{duchon2004boltzmann}, and which meets
constraint programming where a sequence is constrained by an automaton
\cite{pesant2004regular,demassey2006costregul} or by a Markov model \cite{pachet2010markov};
conditioning it on prescribed letter counts is classical within that method as well
\cite{denise2000,denise2010}, and also done by holonomic methods \cite{bertoni2003}.  What is new is that this
set of industrial constraints is of that kind, a regular language conditioned on thresholds and
signed differences of its letter counts rather than on a prescribed composition, together with a
table that carries only what those conditions need.

The paper is organised as follows. In Section~\ref{sec:language} we define the language and the two
kinds of constraint, and in Section~\ref{sec:counting} we give the algorithms, the counts and the
bounds, and verify the counts. In Section~\ref{sec:tree} we compare the construction with branch
and bound, in Section~\ref{sec:sampling} we sample uniformly, and in Section~\ref{sec:search} we
give the recombination operator. Section~\ref{sec:conclusions} discusses what is due to whom and the
limitations, and Section~\ref{sec:final} concludes.

\section{The language of admissible stacking sequences}
\label{sec:language}

\subsection{Words, forbidden factors, and the manufacturing rules}

A stacking sequence of $N$ plies is a word of length $N$ over a finite alphabet $\Sig$ of ply
orientations. We consider the conventional set, its extension by $\pm 30$ and $\pm 60$, and a
finer one,
\begin{align*}
  \Sigfour&=\{0,\pm 45,90\}, \qquad \Sigeight=\{0,\pm 30,\pm 45,\pm 60,90\},\\
  \Sigtwelve&=\{0,\pm 15,\pm 30,\pm 45,\pm 60,\pm 75,90\},
\end{align*}
angles being measured in degrees; $\Sigeight$ is also called the \emph{extended} alphabet.

Recall that a word $u$ is a \emph{factor} of $w$ if $w=xuy$ for some words $x,y$, and that a
set of words $\anti$ is an \emph{antidictionary} for the language
$\lang(\anti)=\{w\in\Sig^{*} : \text{no factor of } w \text{ belongs to } \anti\}$, the words
that \emph{avoid} $\anti$ \cite[Sec.~2]{crochemore1998}, \cite{crochemore2000data}. When $\anti$ is finite, $\lang(\anti)$ is regular
\cite[Sec.~2]{crochemore1998}.

The disorientation between two orientations is
\begin{equation}
  \label{eq:adiff}
  \delta(a,b)=\min\bigl(|a-b| \bmod 180,\; 180-(|a-b| \bmod 180)\bigr),
\end{equation}
and the two manufacturing rules read as follows (Figure~\ref{fig:plies}).

\begin{description}
  \item[\textnormal{(M1) Disorientation.}] No two consecutive plies differ by more than
    $45^{\circ}$, that is, $\delta(w_i,w_{i+1})\le 45$ for every $i$.
  \item[\textnormal{(M2) Contiguity.}] No more than three consecutive plies share the same
    orientation.
\end{description}

\begin{figure*}[!t]
  \centering
  \includegraphics[width=0.85\textwidth]{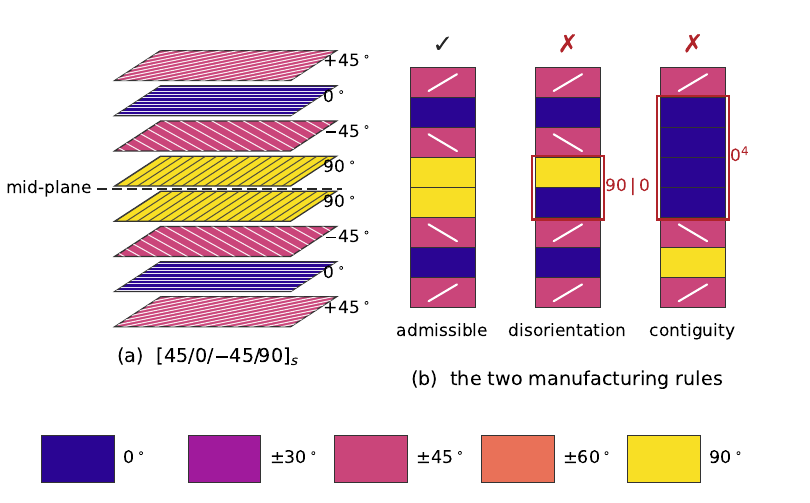}
  \caption{A stacking sequence and the two manufacturing rules. (\textbf{a}) An eight-ply
    symmetric laminate in exploded view, each ply carrying its fibre direction; the mid-plane is
    the axis of the mirror. (\textbf{b}) The laminate of (a) as a column of plies, admissible, and
    two sequences that are not, one against (M1) and one against (M2). Colour encodes $|\theta|$;
    the sign is carried by the fibres and by the lean of the stroke.}
  \label{fig:plies}
\end{figure*}

\begin{proposition}
  \label{prop:anti}
  Rules \textnormal{(M1)} and \textnormal{(M2)} forbid exactly the finite set
  \begin{equation}
    \label{eq:anti}
    \anti=\{ab \;:\; a,b\in\Sig,\ \delta(a,b)>45\}\ \cup\ \{aaaa \;:\; a\in\Sig\},
  \end{equation}
  and the set of stacking sequences satisfying both rules is $\lang(\anti)$. In particular
  $|\anti|=8$ for $\Sigfour$ (four forbidden pairs, namely $0\,90$, $90\,0$, $45\,{-45}$ and
  $-45\,45$, and four forbidden quadruples), $|\anti|=36$ for $\Sigeight$ (twenty-eight
  pairs and eight quadruples) and $|\anti|=72$ for $\Sigtwelve$ (sixty pairs and twelve
  quadruples).
\end{proposition}

\begin{proof}
  Both rules are stated on windows of fixed length, two consecutive plies for (M1) and four
  for (M2), so a sequence violates them if and only if it contains one of the corresponding
  windows as a factor. Enumerating the violating windows gives \eqref{eq:anti}. The
  cardinalities follow by inspection of $\delta$ on $\Sig\times\Sig$.
\end{proof}

\begin{theorem}
  \label{thm:regular}
  The set of stacking sequences satisfying \textnormal{(M1)} and \textnormal{(M2)} is a regular
  language, recognised by a finite automaton with $3|\Sig|$ states besides the initial one, in which the state reached is
  a function of the prefix read.
\end{theorem}

\begin{proof}
  By Proposition~\ref{prop:anti} it suffices to recognise $\lang(\anti)$. Take
  \begin{equation}
    \label{eq:states}
    Q=\{(a,r) \;:\; a\in\Sig,\ 1\le r\le 3\},
  \end{equation}
  where $a$ is the last orientation read and $r$ the length of the run it terminates, with
  transitions
  \begin{equation}
    \label{eq:delta}
    \delta_{Q}\bigl((a,r),b\bigr)=
    \begin{cases}
      (a,r+1) & \text{if } b=a \text{ and } r<3,\\
      (b,1)   & \text{if } b\neq a \text{ and } \delta(a,b)\le 45,\\
      \text{undefined} & \text{otherwise}.
    \end{cases}
  \end{equation}
  Every state is accepting, and the states reachable by a single ply are those with $r=1$; prefixing a state $q_{0}$ for the empty word, with
  $\delta_{Q}(q_{0},a)=(a,1)$, gives a deterministic automaton in the usual sense. A transition
  is undefined exactly when it would complete a factor of $\anti$, by \eqref{eq:anti}: the
  second case of \eqref{eq:delta} rules out the forbidden pairs and the first the quadruples.
  Hence the automaton accepts $\lang(\anti)$.
\end{proof}

Figure~\ref{fig:dfa} shows the automaton for $\Sigfour$, where $|Q|=12$ and there are $32$
transitions; on $\Sigeight$ it has $24$ states and $100$ transitions, on $\Sigtwelve$ $36$ and
$240$. It is, up to its sink state, the automaton that \cite[Sec.~2]{crochemore1998} build from a
finite antidictionary, whose states are the proper prefixes of forbidden words; here those prefixes
are the empty word and the runs of one, two and three identical plies. In Figure~\ref{fig:dfa} the missing cross transitions are
exactly the forbidden pairs, $0\leftrightarrow 90$ and $45\leftrightarrow-45$, so that what
remains is a cycle on the four orientations.

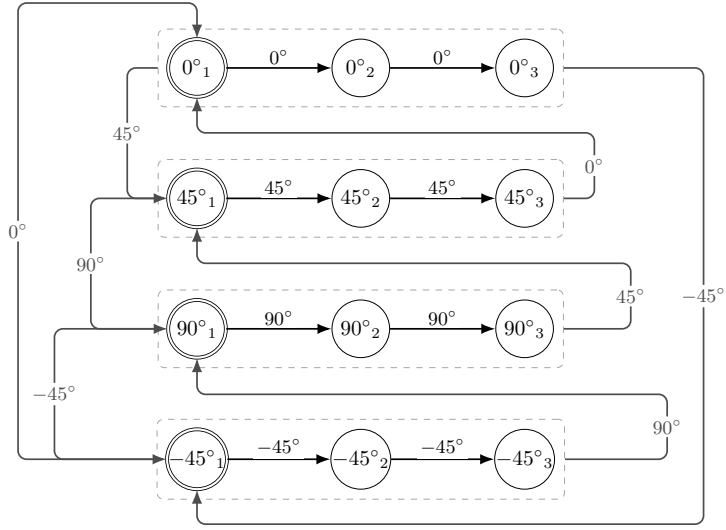
\begin{figure*}[!t]
  \centering
  \resizebox{0.62\textwidth}{!}{%
\begin{tikzpicture}[
    >=Latex, node distance=0pt,
    st/.style={circle, draw, minimum size=9.5mm, inner sep=0pt, font=\small},
    acc/.style={st, double, double distance=1pt},
    grp/.style={rounded corners=3pt, draw=black!35, dashed, inner sep=4.5pt},
    cross/.style={->, thick, black!70, rounded corners=4pt},
    run/.style={->, thick},
    lbl/.style={font=\footnotesize, inner sep=1.5pt, fill=white},
  ]

  \node[acc] (s01) at (0.00,-0.00) {$0^\circ$$_{1}$};
  \node[st] (s02) at (2.70,-0.00) {$0^\circ$$_{2}$};
  \node[st] (s03) at (5.40,-0.00) {$0^\circ$$_{3}$};
  \node[acc] (s11) at (0.00,-2.15) {$45^\circ$$_{1}$};
  \node[st] (s12) at (2.70,-2.15) {$45^\circ$$_{2}$};
  \node[st] (s13) at (5.40,-2.15) {$45^\circ$$_{3}$};
  \node[acc] (s21) at (0.00,-4.30) {$90^\circ$$_{1}$};
  \node[st] (s22) at (2.70,-4.30) {$90^\circ$$_{2}$};
  \node[st] (s23) at (5.40,-4.30) {$90^\circ$$_{3}$};
  \node[acc] (s31) at (0.00,-6.45) {$-45^\circ$$_{1}$};
  \node[st] (s32) at (2.70,-6.45) {$-45^\circ$$_{2}$};
  \node[st] (s33) at (5.40,-6.45) {$-45^\circ$$_{3}$};

  \draw[run] (s01) -- node[lbl, above] {$0^\circ$} (s02);
  \draw[run] (s02) -- node[lbl, above] {$0^\circ$} (s03);
  \draw[run] (s11) -- node[lbl, above] {$45^\circ$} (s12);
  \draw[run] (s12) -- node[lbl, above] {$45^\circ$} (s13);
  \draw[run] (s21) -- node[lbl, above] {$90^\circ$} (s22);
  \draw[run] (s22) -- node[lbl, above] {$90^\circ$} (s23);
  \draw[run] (s31) -- node[lbl, above] {$-45^\circ$} (s32);
  \draw[run] (s32) -- node[lbl, above] {$-45^\circ$} (s33);

  \node[grp, fit=(s01)(s03)] (box0) {};
  \node[grp, fit=(s11)(s13)] (box1) {};
  \node[grp, fit=(s21)(s23)] (box2) {};
  \node[grp, fit=(s31)(s33)] (box3) {};

  \draw[cross] (box0.west) -- (-1.15,-0.00) -- node[lbl, pos=0.5] {$45^\circ$} (-1.15,-2.15) -- (s11.west);
  \draw[cross] (box0.east) -- (8.35,-0.00) -- node[lbl, pos=0.5] {$-45^\circ$} (8.35,-7.52) -- (0,-7.52) -- (s31.south);
  \draw[cross] (box1.east) -- (6.55,-2.15) -- node[lbl, pos=0.5] {$0^\circ$} (6.55,-1.07) -- (0,-1.07) -- (s01.south);
  \draw[cross] (box1.west) -- (-1.75,-2.15) -- node[lbl, pos=0.5] {$90^\circ$} (-1.75,-4.30) -- (s21.west);
  \draw[cross] (box2.east) -- (7.15,-4.30) -- node[lbl, pos=0.5] {$45^\circ$} (7.15,-3.22) -- (0,-3.22) -- (s11.south);
  \draw[cross] (box2.west) -- (-2.35,-4.30) -- node[lbl, pos=0.5] {$-45^\circ$} (-2.35,-6.45) -- (s31.west);
  \draw[cross] (box3.west) -- (-2.95,-6.45) -- node[lbl, pos=0.5] {$0^\circ$} (-2.95,1.07) -- (0,1.07) -- (s01.north);
  \draw[cross] (box3.east) -- (7.75,-6.45) -- node[lbl, pos=0.5] {$90^\circ$} (7.75,-5.38) -- (0,-5.38) -- (s21.south);

\end{tikzpicture}
}
  \caption{The automaton of the manufacturing rules on $\Sigfour$. A state $\theta_{r}$ is the pair
    (last orientation, run length). Horizontal arrows repeat the last ply and advance the run,
    which stops at three; returning arrows change orientation and reset the run to one; the three states of an orientation
    are drawn as one dashed box, and each returning arrow stands for three transitions. Double circles mark the states
    with run one; the initial state $q_{0}$ is not drawn.}
  \label{fig:dfa}
\end{figure*}

\subsection{The design guidelines as Parikh constraints}

The remaining guidelines do not forbid factors: they constrain how many plies of each
orientation the laminate contains. Let $\parikh{w}$ denote the \emph{Parikh vector} of $w$, whose
entry $c_{a}(w)$ counts the occurrences of $a$ in $w$. For a laminate of $N$ plies they read
\begin{align}
  \label{eq:balance}
  c_{+\theta}(w)&=c_{-\theta}(w) && \text{for every } \theta \text{ with } \pm\theta\in\Sig,
  &&\text{(balance)}\\
  \label{eq:tenpct}
  c_{a}(w)&\ge 0.1\,N && \text{for every } a\in\{0,45,-45,90\},
  &&\text{(10\% rule)}
\end{align}
and, where it is required, mirror symmetry about the mid-plane. The constructions below take each guideline as an option.

\begin{proposition}
  \label{prop:notregular}
  The words of all lengths over $\Sig\supseteq\Sigfour$ that satisfy \textnormal{(M1)},
  \textnormal{(M2)} and \eqref{eq:balance} do not form a regular language, and neither do those
  that satisfy \eqref{eq:tenpct} as well.
\end{proposition}

\begin{proof}
  Intersect with the regular language $R=(45\;90)^{*}(-45\;0)^{*}$. In a word of $R$ consecutive
  plies differ by $45^{\circ}$ and no orientation is repeated, so (M1) and (M2) hold, and
  $c_{45}=c_{90}=n$, $c_{-45}=c_{0}=k$ for some $n,k\ge 0$. Balance holds exactly when $n=k$, and
  then each of the four directions is a quarter of the plies, so \eqref{eq:tenpct} holds too. The
  intersection is $\{(45\;90)^{n}(-45\;0)^{n} : n\ge 0\}$ in both cases, which is not regular by
  the pumping lemma; since regular languages are closed under intersection, neither set is regular.
\end{proof}

Since \eqref{eq:balance} and \eqref{eq:tenpct} depend on $w$ only through $\parikh{w}$,
they are carried exactly by a dynamic program over pairs (automaton state, partial Parikh vector),
developed in Section~\ref{sec:counting}. The same construction accommodates any constraint that depends on
the word only through its Parikh vector, given counters that carry what that constraint needs.

\subsection{Symmetric and non-symmetric laminates}
\label{sec:halfstack}

A symmetric laminate of $N=2m$ plies is written
$w=h\,h^{R}$, where $h$ is the \emph{half-stack} of $m$ plies and $h^{R}$ its reversal.

\begin{lemma}
  \label{lem:mirror}
  Let $N=2m$ and $w=h\,h^{R}$. Then $w$ satisfies \textnormal{(M1)} and \textnormal{(M2)} if and
  only if $h$ does and $h$ ends with a run of length one.
\end{lemma}

\begin{proof}
  The factors of $w$ that are not factors of $h$ or of $h^{R}$ are those straddling the
  mid-plane, where the last ply of $h$ is immediately followed by its own copy. Rule (M1) is
  satisfied there because $\delta(a,a)=0$. Rule (M2) is satisfied there if and only if the run
  ending $h$, of length $r$, does not become a run of length $2r>3$, that is if and only if
  $r=1$. Reversal preserves both rules, since $\anti$ is closed under reversal.
\end{proof}

\begin{remark}
For odd $N$ the mid-plane carries a single ply that its own reflection does not duplicate, so
the run condition of Lemma~\ref{lem:mirror} is replaced by one on the run that ply terminates.
The Parikh vector becomes $2\parikh{h}+e_{c}$, where $c$ is the central ply, so balance forces
$c\in\{0,90\}$ and the thresholds depend on $c$. We treat only even $N$ below.
\end{remark}

\begin{lemma}
  \label{lem:parikh-half}
  With $w=h\,h^{R}$ we have $\parikh{w}=2\,\parikh{h}$, so \eqref{eq:balance} and
  \eqref{eq:tenpct} hold for $w$ if and only if $c_{+\theta}(h)=c_{-\theta}(h)$ for every
  $\theta$ and $2c_{a}(h)\ge 0.1\,N$ for every $a\in\{0,45,-45,90\}$.
\end{lemma}

\begin{proof}
  Immediate, since $h^{R}$ is a permutation of $h$.
\end{proof}

Lemma~\ref{lem:mirror} turns symmetry into a choice of accepting states, those with run one
(the double circles of Figure~\ref{fig:dfa}), and Lemma~\ref{lem:parikh-half} halves the
counters.

A non-symmetric laminate is read whole; ``non-symmetric'' means here that symmetry is not
imposed, so the class contains the symmetric laminates, a share below $2\times 10^{-6}$ of it at
every ply count tabulated. Rules (M1) and (M2) apply to the word of $N$ plies,
balance and the $10\%$ rule to its Parikh vector, and every state is accepting. We write $L$ for the length of the word the
automaton reads, $L=m=N/2$ for a symmetric laminate and $L=N$ otherwise, and $t$ for the smallest
value of a counter of that word that satisfies the $10\%$ rule, that is the least integer with
$2t\ge 0.1N$ in the symmetric case and $t\ge 0.1N$ otherwise. Everything that follows holds for
both unless stated otherwise.

\begin{definition}
  \label{def:admissible}
  Fix $N$ and the guidelines imposed, a subset $G$ of symmetry, balance \eqref{eq:balance} and the
  $10\%$ rule \eqref{eq:tenpct}. A laminate is \emph{admissible} if it satisfies \textnormal{(M1)},
  \textnormal{(M2)} and every guideline in $G$, and \emph{compliant} if it is admissible with
  balance and the $10\%$ rule in $G$, the setting of Sections~\ref{sec:sampling} and~\ref{sec:search}. A word read by the automaton, $h$ or $w$, is admissible or compliant when its
  laminate is. A \emph{completion} of a word $u$ of length $j<L$ is a word $v$ of length $L-j$ such
  that $uv$ is admissible. With $G$ empty the admissible words are the \emph{manufacturing
  language} $\lang(\anti)\cap\Sig^{L}$.
\end{definition}

The statements of Section~\ref{sec:counting} hold for every $G$.

\section{Counting}
\label{sec:counting}

\subsection{The transfer matrix}

The term is meant in its combinatorial sense \cite[Ex.~V.11]{flajolet2009analytic}, unrelated to
the transfer-matrix method of laminate mechanics.

Let $A$ be the $|Q|\times|Q|$ adjacency matrix of the automaton of Theorem~\ref{thm:regular},
$A_{qq'}=1$ if $\delta_{Q}(q,b)=q'$ for some $b\in\Sig$ and $0$ otherwise; since the state
reached determines the ply read, walks and words correspond one to one. By
Lemma~\ref{lem:mirror} the admissible half-stacks of $m$ plies are the walks of length $m-1$
that start and end in a state with run one, so their number is
\begin{equation}
  \label{eq:transfer}
  \bigl|\lang_{m}\bigr| \;=\; \boldsymbol{u}^{\top}A^{\,m-1}\boldsymbol{u},
  \qquad
  u_{q}=\begin{cases}1 & \text{if } q=(a,1) \text{ for some } a\in\Sig,\\ 0 &\text{otherwise.}\end{cases}
\end{equation}
For a non-symmetric laminate the count is $\boldsymbol{u}^{\top}A^{\,N-1}\boldsymbol{1}$. Equation~\eqref{eq:transfer} is an exact
integer computation whose cost grows with the number of digits of the count and not with the
count, and produces no sequence.

\subsection{The dynamic program on the Parikh counters}
\label{sec:dp}

We work on \emph{cells} $(q,\kappa)$, where $q$ is an automaton state and $\kappa$
a vector of counters holding what the guidelines need to know of the Parikh vector read so
far:
\begin{itemize}
  \item for each principal direction $a\in\{0,\pm 45,90\}\cap\Sig$, the count $\min(c_{a},t)$,
    since beyond $t$ the $10\%$ rule is indifferent;
  \item for each pair $\pm\theta\subseteq\Sig$ with $\theta>0$, when balance is imposed, the
    difference $d_{\theta}=c_{+\theta}-c_{-\theta}$.
\end{itemize}
Reading a ply $b$ updates $\kappa$ to a vector $\kappa\oplus b$ that depends on $\kappa$ and $b$
only. Let $\operatorname{need}(\kappa)$ be the least number of plies that, appended
to a word with counters $\kappa$, satisfy every imposed guideline, ignoring the automaton: the
deficits $t-\min(c_{a},t)$ of the thresholded directions that belong to no pair, $|d_{\theta}|$ for
the pairs without threshold, and $2\max(c_{45},c_{-45},t)-c_{45}-c_{-45}$ for $\pm 45$ when both are imposed, which the
capped counts and $d_{45}$ determine; a guideline that is not imposed contributes nothing, and with
the $10\%$ rule off $t=0$. A cell is \emph{accepting} when $\operatorname{need}(\kappa)=0$
and, for a symmetric laminate, $q$ has run one (Lemma~\ref{lem:mirror}).

\begin{lemma}
  \label{lem:cells}
  Two words of the same length that satisfy \textnormal{(M1)} and \textnormal{(M2)} and reach the
  same cell have the same completions (Definition~\ref{def:admissible}). Moreover
  $\operatorname{need}(\kappa\oplus b)\ge\operatorname{need}(\kappa)-1$ for every $b$, so a word of
  length $j$ whose cell has $\operatorname{need}(\kappa)>L-j$ has no completion, and neither has
  any of its extensions.
\end{lemma}

\begin{proof}
  By \eqref{eq:delta} the continuations readable after a prefix depend on it only through $q$, and
  whether a completion satisfies the guidelines depends on the prefix only through $\kappa$, which
  $\oplus$ updates without reference to anything else. A ply raises one entry of the Parikh
  vector by one, so it moves a single term of $\operatorname{need}$, that of its direction or of
  its pair, and lowers it by at most one; the remaining $L-j$ plies can therefore close at most
  $L-j$ units of deficit.
\end{proof}

The program has two passes (Algorithm~\ref{alg:table}). The forward pass
generates the reachable cells level by level, discarding by $\operatorname{need}$; the backward pass
assigns to each cell $F_{j}(q,\kappa)$, the number of admissible completions of any prefix of length
$j$ that reaches it. The number of admissible words is the sum of $F_{1}$ over the first level. The
threshold $t$ is computed in integer arithmetic, as the least integer with $10t\ge N$, or
$20t\ge N$ for a half-stack.

\begin{algorithm}[!htbp]
  \caption{\textsc{Automaton}$(\Sig)$: the manufacturing rules (M1)--(M2).}
  \label{alg:dfa}
  \begin{algorithmic}[1]
    \State $Q\gets\{(a,r) : a\in\Sig,\ 1\le r\le 3\}$
    \For{$(a,r)\in Q$ \textbf{and} $b\in\Sig$}
      \If{$b=a$ \textbf{and} $r<3$} \State $\delta_{Q}((a,r),b)\gets(a,r+1)$
        \Comment{extend the run}
      \ElsIf{$b\ne a$ \textbf{and} $\delta(a,b)\le 45$} \State $\delta_{Q}((a,r),b)\gets(b,1)$
        \Comment{open a new run}
      \EndIf \Comment{otherwise undefined: $b$ completes a factor of $\anti$}
    \EndFor
    \State \Return $(Q,\delta_{Q})$
  \end{algorithmic}
\end{algorithm}

\begin{algorithm}[!htbp]
  \caption{\textsc{Table}$(\Sig,N,\textit{options})$: completion counts on the cells.}
  \label{alg:table}
  \begin{algorithmic}[1]
    \Require options: symmetric or not, balance on or off, $10\%$ rule on or off
    \State $(Q,\delta_{Q})\gets\textsc{Automaton}(\Sig)$;\quad
      $L\gets N/2$ if symmetric, else $N$
    \State $C_{1}\gets\{((a,1),\,\mathbf{0}\oplus a) : a\in\Sig,\ \operatorname{need}(\mathbf{0}\oplus a)\le L-1\}$
    \For{$j=1,\dots,L-1$} \Comment{forward: reachable cells}
      \For{$(q,\kappa)\in C_{j}$ \textbf{and} $b$ with $\delta_{Q}(q,b)$ defined}
        \If{$\operatorname{need}(\kappa\oplus b)\le L-j-1$}
          \State add $(\delta_{Q}(q,b),\kappa\oplus b)$ to $C_{j+1}$
        \EndIf
      \EndFor
    \EndFor
    \For{$(q,\kappa)\in C_{L}$}
      \State $F_{L}(q,\kappa)\gets 1$ if $\operatorname{need}(\kappa)=0$ and ($q$ has run one or not symmetric), else $0$
    \EndFor
    \For{$j=L-1,\dots,1$} \Comment{backward: completions}
      \For{$(q,\kappa)\in C_{j}$}
        \State $F_{j}(q,\kappa)\gets\sum_{b} F_{j+1}(\delta_{Q}(q,b),\kappa\oplus b)$
          \Comment{over defined $\delta_{Q}$; absent cells count $0$}
      \EndFor
    \EndFor
    \State \Return $F$, and $|\lang|=\sum_{c\in C_{1}}F_{1}(c)$
  \end{algorithmic}
\end{algorithm}

\begin{algorithm}[!htbp]
  \caption{\textsc{Sample}$(F)$ and \textsc{Completable}$(F,p)$.}
  \label{alg:sample}
  \begin{algorithmic}[1]
    \If{$|\lang|=0$} \State \Return none \Comment{the empty set, certified} \EndIf
    \State draw $c\in C_{1}$ with probability $F_{1}(c)/|\lang|$;\quad $w\gets$ the ply of $c$
    \For{$j=1,\dots,L-1$}
      \State draw $b$ with probability $F_{j+1}(\delta_{Q}(q,b),\kappa\oplus b)/F_{j}(q,\kappa)$, where $c=(q,\kappa)$
      \State $c\gets(\delta_{Q}(q,b),\kappa\oplus b)$;\quad append $b$ to $w$
    \EndFor
    \State \Return $w$ \Comment{uniform on $\lang$; mirror it if symmetric}
    \Statex
    \Function{Completable}{$F,p$} \Comment{$p$ a prefix of length $j$}
      \If{$j=0$} \Return $|\lang|>0$ \EndIf
      \State walk $p$ through $\delta_{Q}$ from $q_{0}$ and $\oplus$ from $\mathbf{0}$
      (Theorem~\ref{thm:regular}); \Return false if a transition is undefined
      \State \Return $F_{j}(q,\kappa)>0$ for the cell $(q,\kappa)$ reached
    \EndFunction
  \end{algorithmic}
\end{algorithm}

\begin{theorem}
  \label{thm:poly}
  Let $\Sig$ be fixed, with $p$ principal directions and $k$ sign pairs. Algorithm~\ref{alg:table}
  computes the number $|\lang|$ of admissible words of length $L$ exactly, in
  \[ O\bigl(|\Sig|\,|Q|\,L\,(t+1)^{p}(2L+1)^{k}\bigr) \]
  arithmetic operations on integers of $O(L\log|\Sig|)$ bits, hence in time polynomial in $N$.
  After it, Algorithm~\ref{alg:sample} draws a uniformly distributed admissible word in $O(|\Sig|L)$
  operations and decides whether a prefix of length $j$ is completable in $O(j)$.
\end{theorem}

\begin{proof}
  Correctness: by Lemma~\ref{lem:cells} the prefixes that reach a cell share their completions, so
  $F_{j}$ is well defined, the recursion of the backward pass counts each admissible completion
  exactly once by its next ply, and the discarded cells contribute nothing. Each level holds at most
  $|Q|$ states times $(t+1)^{p}$ capped counts times $(2L+1)^{k}$ differences, and each cell has at
  most $|\Sig|$ successors. A count never exceeds $|\Sig|^{L}$. The walk of
  Algorithm~\ref{alg:sample} chooses each ply with probability equal to the share of completions it
  keeps, so the product telescopes to $1/|\lang|$ for every admissible word.
\end{proof}

The operations counted are additions and comparisons of integers of that size; a draw takes a
uniform integer below $F_{j}$, and the cells are kept in hash tables with constant expected access.

With the discarding by $\operatorname{need}$, the largest level at $N=48$ on $\Sigeight$ holds
$213\,202$ cells. When the distribution over compositions
is wanted, as below, the counters are carried
uncapped, at a cost in cells.

\subsection{Cost of the design guidelines}

\begin{table*}[!t]
  \centering
  \caption{Exact number of admissible symmetric half-stacks of $N/2$ plies on the three alphabets:
    (a) under the manufacturing rules alone, (b) under the rules together with the design
    guidelines, with the ratios between alphabets. In the $\Sigtwelve$ columns values of
    $10^{7}$ and above are rounded to three significant figures; the deposited code returns the integers.}
  \label{tab:counts}
  \begin{tabular}{r rrr}
    \toprule
    & \multicolumn{3}{c}{(a) manufacturing rules only}\\
    \cmidrule(lr){2-4}
    $N$ & $\Sigfour$ & $\Sigeight$ & $\Sigtwelve$\\
    \midrule
    24 & $376\,704$ & $102\,102\,368$ & $1.99\times 10^{10}$ \\
    40 & $1\,988\,997\,120$ & $17\,885\,340\,482\,800$ & $1.13\times 10^{17}$ \\
    44 & $16\,954\,798\,592$ & $365\,900\,128\,074\,196$ & $5.49\times 10^{18}$ \\
    48 & $144\,527\,708\,160$ & $7\,485\,622\,309\,145\,552$ & $2.68\times 10^{20}$ \\
    54 & $3\,596\,986\,742\,784$ & $692\,669\,292\,246\,653\,848$ & $9.11\times 10^{22}$ \\
    \bottomrule
  \end{tabular}

  \vspace{1.2em}
  \begin{tabular}{r rrr rr}
    \toprule
    & \multicolumn{3}{c}{(b) rules and guidelines} & & \\
    \cmidrule(lr){2-4}
    $N$ & $\Sigfour$ & $\Sigeight$ & $\Sigtwelve$ & $\Sigeight/\Sigfour$ & $\Sigtwelve/\Sigeight$\\
    \midrule
    24 & $29\,560$ & $182\,020$ & $782\,348$ & $6.2\times$ & $4.3\times$ \\
    40 & $175\,964\,268$ & $39\,709\,914\,088$ & $6.32\times 10^{12}$ & $226\times$ & $159\times$ \\
    44 & $1\,302\,502\,076$ & $398\,548\,747\,372$ & $3.46\times 10^{13}$ & $306\times$ & $87\times$ \\
    48 & $11\,236\,189\,644$ & $8\,999\,374\,621\,960$ & $2.18\times 10^{15}$ & $801\times$ & $243\times$ \\
    54 & $277\,909\,460\,780$ & $881\,013\,334\,451\,080$ & $8.49\times 10^{17}$ & $3\,170\times$ & $964\times$ \\
    \bottomrule
  \end{tabular}
\end{table*}

At $N=48$ the guidelines divide the language by $13$ on $\Sigfour$ and by $832$ on $\Sigeight$ (Table~\ref{tab:counts}). On $\Sigeight$ balance acts on three pairs instead of one, and the $10\%$ rule demands that $0$,
$\pm 45$ and $90$ all be present in a laminate whose plies may also be spent on $\pm 30$ and
$\pm 60$, which do not help to satisfy it.

Hence the size advantage of the extended alphabet falls from $51\,794\times$ on the
manufacturing language to $801\times$ under the guidelines. On $\Sigtwelve$ the guidelines divide the language by
$1.2\times 10^{5}$, and its advantage over $\Sigeight$ falls from $35\,700\times$ to
$243\times$.

\subsection{Non-symmetric and unbalanced laminates}
\label{sec:nonsym}

An unbalanced laminate couples
extension with shear and strengthens the coupling of bending with twisting, a non-symmetric one
couples extension with bending \cite{shirk1986aeroelastic}. Algorithm~\ref{alg:table} counts every class with the
corresponding options (Section~\ref{sec:halfstack}); for the non-symmetric balanced class, which
carries a difference for every sign pair over the whole length, we count by meeting in the middle.

\begin{proposition}
  \label{prop:mitm}
  Let $N$ be even and let $G(q,\kappa)$ be the number of words of $N/2$ plies satisfying
  \textnormal{(M1)}--\textnormal{(M2)} that reach the cell $(q,\kappa)$, obtained by carrying these
  counts through the forward pass of Algorithm~\ref{alg:table} for non-symmetric words of
  length $N$. Then the number of admissible non-symmetric laminates of $N$ plies is
  \[
    \sum G(q,\kappa)\,G(q',\kappa'),
  \]
  over the pairs of cells such that the junction is admissible (if $q=(a,r)$ and $q'=(b,r')$, then
  $r+r'\le 3$ when $a=b$ and $\delta(a,b)\le 45$ otherwise), every sign-pair difference of $\kappa'$
  is the opposite of that of $\kappa$, and $\min(c_{a},t)+\min(c'_{a},t)\ge t$ for every principal
  direction $a$.
\end{proposition}

\begin{proof}
  Write the laminate as $uv$ with $|u|=|v|=N/2$. The antidictionary \eqref{eq:anti} is closed
  under reversal, and the Parikh constraints do not depend on the order, so $v$ is admissible
  exactly when its reversal $v^{R}$ is an admissible prefix of the same length; the state reached by
  $v^{R}$ is the pair (first ply of $v$, length of its first run). The factors of $uv$ that are
  factors of neither $u$ nor $v$ straddle the junction; they break (M1) exactly when
  $\delta(a,b)>45$, and (M2) exactly when $a=b$ and the two runs together exceed three. The Parikh
  vector of $uv$ is the sum of those of $u$ and $v$, so balance holds iff the differences are
  opposite, and a threshold holds iff the sum of the counts reaches $t$, which the capped counts
  decide. The pruning of the forward pass removes only prefixes that no completion of length $N$
  can save (Lemma~\ref{lem:cells}), hence none that the sum needs.
\end{proof}

Proposition~\ref{prop:mitm} needs only the middle level, one table serving both halves. With the counts of each half exact and the sum
taken modulo four primes of $61$ bits and rebuilt by the Chinese remainder theorem, it reproduces
all $56$ non-symmetric counts obtained with Algorithm~\ref{alg:table}, and reaches $N=54$ on
$\Sigtwelve$ with $1.1\times 10^{9}$ cells at the middle level (Table~\ref{tab:classes}).

\begin{table*}[!t]
  \centering
  \caption{Exact number of admissible laminates of $N$ plies under the manufacturing rules and the
    $10\%$ rule, for the three classes other than symmetric and balanced (Table~\ref{tab:counts}):
    symmetric unbalanced, non-symmetric balanced, non-symmetric unbalanced.
    Values of $10^{7}$ and above are rounded to three significant figures; the deposited code
    returns the integers.}
  \label{tab:classes}
  \begin{tabular}{l r rrr}
    \toprule
    & & symmetric & \multicolumn{2}{c}{non-symmetric}\\
    \cmidrule(lr){3-3}\cmidrule(lr){4-5}
    & $N$ & unbalanced & balanced & unbalanced\\
    \midrule
    $\Sigfour$ & 24 & $147\,112$ & $1.62\times 10^{10}$ & $1.63\times 10^{11}$ \\
     & 40 & $1.69\times 10^{9}$ & $3.91\times 10^{17}$ & $5.64\times 10^{18}$ \\
     & 44 & $1.18\times 10^{10}$ & $2.69\times 10^{19}$ & $4.02\times 10^{20}$ \\
     & 48 & $1.13\times 10^{11}$ & $1.90\times 10^{21}$ & $3.01\times 10^{22}$ \\
     & 54 & $3.12\times 10^{12}$ & $1.11\times 10^{24}$ & $1.87\times 10^{25}$ \\
    \midrule
    $\Sigeight$ & 24 & $4\,128\,968$ & $1.11\times 10^{13}$ & $1.64\times 10^{15}$ \\
     & 40 & $6.05\times 10^{12}$ & $2.51\times 10^{23}$ & $1.32\times 10^{26}$ \\
     & 44 & $4.15\times 10^{13}$ & $7.75\times 10^{25}$ & $3.98\times 10^{28}$ \\
     & 48 & $1.32\times 10^{15}$ & $3.32\times 10^{28}$ & $2.30\times 10^{31}$ \\
     & 54 & $1.98\times 10^{17}$ & $2.23\times 10^{32}$ & $1.72\times 10^{35}$ \\
    \midrule
    $\Sigtwelve$ & 24 & $3.33\times 10^{7}$ & $2.49\times 10^{15}$ & $1.57\times 10^{18}$ \\
     & 40 & $4.27\times 10^{15}$ & $4.80\times 10^{28}$ & $1.91\times 10^{32}$ \\
     & 44 & $1.31\times 10^{16}$ & $2.67\times 10^{31}$ & $9.63\times 10^{34}$ \\
     & 48 & $1.35\times 10^{18}$ & $1.01\times 10^{35}$ & $5.68\times 10^{38}$ \\
     & 54 & $1.14\times 10^{21}$ & $3.69\times 10^{39}$ & $2.32\times 10^{43}$ \\
    \bottomrule
  \end{tabular}
\end{table*}

At $N=48$ the balanced non-symmetric laminates outnumber the symmetric ones by
$1.7\times 10^{11}$ on $\Sigfour$, $3.7\times 10^{15}$ on $\Sigeight$ and $4.6\times 10^{19}$ on
$\Sigtwelve$. Dropping balance multiplies the non-symmetric set by $16$, $691$ and $5\,648$. The symmetric
unbalanced class is $10.0$,
$147$ and $619$ times the symmetric balanced one. Every algorithm of this paper applies to
these sets wherever their full table fits in memory: sampling and the completion test use it as it
is, and the crossover of Section~\ref{sec:search} is used without the run-one condition that a
symmetric laminate imposes at its mid-plane (Lemma~\ref{lem:mirror}). The balanced non-symmetric
class on $\Sigtwelve$ is counted beyond that point, from $N=40$, by Proposition~\ref{prop:mitm}
alone, and is not sampled here.

\subsection{Ceilings on the composition}

\begin{proposition}
  \label{prop:ceiling}
  Let $t$ be the smallest half-stack counter satisfying the $10\%$ rule, and let $h$ be a
  compliant half-stack of $m$ plies. Then
  \begin{equation}
    \label{eq:ceiling}
    c_{a}(h)\ \le\
    \begin{cases}
      m-3t, & a\in\{0,90\},\\[2pt]
      \bigl\lfloor (m-2t)/2 \bigr\rfloor, & a\in\{45,-45\},\\[2pt]
      \bigl\lfloor (m-4t)/2 \bigr\rfloor, & a \text{ in a sign pair without threshold}.
    \end{cases}
  \end{equation}
\end{proposition}

\begin{proof}
  By Lemma~\ref{lem:parikh-half} the $10\%$ rule reads $c_{b}(h)\ge t$ for each of the four
  directions $b\in\{0,45,-45,90\}$, and balance reads $c_{+\theta}(h)=c_{-\theta}(h)$. Fix $a$ and
  bound the plies that cannot be spent on it. If $a\in\{0,90\}$, which have no distinct partner,
  the other three thresholded directions consume at least $3t$, whence $c_{a}\le m-3t$. If
  $a\in\{45,-45\}$, balance forces $c_{-a}=c_{a}$, so the pair consumes $2c_{a}$, while $0$ and
  $90$ consume at least $2t$; hence $2c_{a}\le m-2t$. If $a$ belongs to a sign pair that
  carries no threshold ($\pm 30$ and $\pm 60$, and on $\Sigtwelve$ also $\pm 15$ and $\pm 75$), the pair again consumes $2c_{a}$ while all four thresholded directions
  consume at least $4t$, hence $2c_{a}\le m-4t$.
\end{proof}

At $N=48$, where $m=24$ and $t=3$, \eqref{eq:ceiling} gives $c_{0}\le 15$. That bound is attained: the dynamic program
finds $840$ compliant half-stacks with exactly $15$, on $\Sigfour$ and on $\Sigeight$.

It is not tight in general. At $N=20$ on $\Sigfour$ the bound gives $c_{0}\le 7$, while
exhaustive enumeration shows the true maximum is $5$. The bound ignores adjacency: a half-stack rich in $0^{\circ}$ must break its runs every three
plies, the only orientations adjacent to $0^{\circ}$ are $\pm 45$, and the $90^{\circ}$ ply that
the $10\%$ rule demands cannot sit next to a $0^{\circ}$ one. The exact maximum needs the joint
state of automaton and counters, which the dynamic program carries.

Against exhaustive enumeration for every orientation at $N=12,16,20$
on $\Sigfour$ and $N=12,14$ on $\Sigeight$, \eqref{eq:ceiling} is attained everywhere
except at $N=20$ on $\Sigfour$, where the maxima fall short of it also for
$90$ ($5$ against $7$) and for $\pm 45$ ($3$ against $4$).

On $\Sigeight$ at $N=48$ the compliant language has mode $c_{0}=3$ and mean $4.42$, and
$0.00058\%$ of it lies at $c_{0}\ge 12$, half the half-stack.

\subsection{Verification of the counts}
\label{sec:verification}

Internally, we compared the counts with exhaustive enumeration filtered by rule checkers written
separately from the dynamic program, on totals and on the whole distribution by composition. The nine cases are
$N=8,12,\dots,24$ on $\Sigfour$ and $N=8,10,\dots,14$ on $\Sigeight$, and all coincide. The ten
manufacturing counts of Table~\ref{tab:counts}(a) on $\Sigfour$ and $\Sigeight$ agree with the
matrix powers of \eqref{eq:transfer}, computed in exact integer arithmetic. We checked the general program of Algorithm~\ref{alg:table},
with capped counters, separately: against exhaustive enumeration on $80$ cases that
cover the three alphabets, symmetric and non-symmetric words, balance on and off and the $10\%$ rule
on and off; against every $\Sigfour$ and $\Sigeight$ entry of Table~\ref{tab:counts}; by $2000$ draws of
Algorithm~\ref{alg:sample} in each of three settings, all compliant; and by \textsc{Completable} on
every one of the $\sum_{j\le 8}4^{j}$ prefixes at $N=16$ on $\Sigfour$, with no disagreement.

Externally, Table~II-1 of \cite{irisarri2009} tabulates the number of realisable stacking
sequences on $\{0,\pm 45,90\}$ rule by rule at $4$, $8$ and $16$ plies. We recover five of its six rule columns at both $4$ and $8$ plies, and with them the
definitions he used. Two differences forbid a column-by-column comparison with our own counts.
His rules are symmetry (Sym.), balance of $\pm 45$ (Eq.), contiguity (Group.), disorientation
(Des., our (M1)), no $0^{\circ}$ ply on either outer surface (TolDom.) and a $10\%$ rule. His
contiguity rule allows two identical consecutive plies where ours allows three; and his
$10\%$ rule is not a counting rule but a membrane-stiffness criterion, chosen to remain applicable
to arbitrary orientations and, as he notes, less restrictive for thin laminates.

\subsection{Completion of a published table}
\label{sec:completing}

At sixteen plies Table~II-1 leaves five of its six single-rule entries \emph{uncomputed}
\cite[p.~59]{irisarri2009}: only the symmetry entry, which it deduces from the eight-ply
enumeration, is filled. Four cumulative entries are also uncomputed; each combines rules with the
$10\%$ criterion. Each of the five combinatorial rules is counted
by a transfer matrix or in closed form, so
we fill four of the five blank cells in the sixteen-ply column of Table~\ref{tab:irisarri}; the fifth, the $10\%$ row, rests on a reading of a membrane-stiffness criterion that the thesis does not spell out and we cannot reconstruct.

\begin{table*}[!t]
  \centering\small
  \caption{Table~II-1 of \cite{irisarri2009} completed. The columns at four and eight plies are the
    published values, which our methods reproduce exactly except in the $10\%$ row; at sixteen plies the symmetry row is
    published and the other entries are the cells the table leaves uncomputed.}
  \label{tab:irisarri}
  \begin{tabular}{l r r r l}
    \toprule
    rule & $4$ plies & $8$ plies & $16$ plies & closed form \\
    \midrule
    Sym.    & $16$  & $256$    & $65\,536$          & $4^{n/2}$, $n$ even \\
    Eq.     & $70$  & $12\,870$ & $601\,080\,390$   & $\binom{2n}{n}$ \\
    Group.  & $228$ & $47\,088$ & $2\,010\,036\,708$ & $a_{n}=3(a_{n-1}+a_{n-2})$ \\
    Des.    & $108$ & $8\,748$  & $57\,395\,628$    & $4\cdot 3^{\,n-1}$ \\
    TolDom. & $144$ & $36\,864$ & $2\,415\,919\,104$ & $9\cdot 4^{\,n-2}$, $n\ge 2$ \\
    $10\%$  & $192$ & $62\,720$ & \multicolumn{2}{l}{not reproducible: see text} \\
    \bottomrule
  \end{tabular}
\end{table*}

Every entry at sixteen plies comes from a computation that first returns the published values at
four and eight plies; the four other than symmetry, which is a formula, also agree with an
independent exhaustive enumeration at eight plies.

Four of the five combinatorial rows are closed forms, and the
fifth, contiguity, satisfies a two-term linear recurrence,
\begin{equation}
  \label{eq:grouprec}
  a_{n}=3\bigl(a_{n-1}+a_{n-2}\bigr),\qquad a_{1}=4,\ a_{2}=16,
\end{equation}
which returns the published values at four and eight plies \cite{oeis_a123620}.
Balance, for its part, is the central binomial coefficient:
\begin{equation}
  \label{eq:eqbinom}
  \sum_{k=0}^{\lfloor n/2\rfloor}\binom{n}{k}\binom{n-k}{k}2^{\,n-2k}=\binom{2n}{n}.
\end{equation}
It is the coefficient identity
$\sum_{k}\binom{n}{k}\binom{n-k}{k}2^{\,n-2k}=[x^{n}](1+2x+x^{2})^{n}=[x^{n}](1+x)^{2n}$, and the
values are those of \cite{oeis_a000984}.

The criterion behind that row requires the extensional stiffnesses $A_{11}$, $A_{22}$ and $A_{66}$ of
classical lamination theory \cite{jones1999mechanics} to be at least those of a virtual
quasi-isotropic laminate $[45/0/-45/90]_{s}$ of $40\%$ the thickness. With plies of one thickness it
depends on the sequence only through its counts and enters the program as three linear inequalities on
uncapped counters; as stated, it gives $224$ and $64\,128$ admissible sequences at $4$ and $8$ plies,
against the published $192$ and $62\,720$, for any carbon-epoxy ply.

\section{Branch and bound}
\label{sec:tree}

A branch and bound for the stacking sequence \cite{kim2005gust} assigns plies one at a time from
the outer surface towards the mid-plane, since the outer plies weigh most in bending. Each node of
its tree is a prefix, discarded when a bound on the objective over its completions
cannot beat the incumbent \cite{morrison2016bnb}. The manufacturing rules are local,
so a node that breaks one is discarded on the spot. The guidelines are not, and a prefix that obeys
the rules may have no admissible completion; the tree learns it only after generating and
rejecting every leaf below, unless the bound happens to discard the branch first for another
reason.

The prefixes of length $j$ that obey the rules
number $T_{\mathrm{loc}}(j)=\boldsymbol{u}^{\top}A^{\,j-1}\boldsymbol{1}$, by the transfer matrix
of Section~\ref{sec:counting}. Those with at least one admissible completion number
$T_{\mathrm{ex}}(j)=\sum G_{j}(q,\kappa)$ over the cells of level $j$ with $F_{j}(q,\kappa)>0$,
where $G_{j}$ counts the prefixes that reach a cell and is obtained by carrying counts through the
forward pass of Algorithm~\ref{alg:table}. Their difference is the part of a locally checked tree
that leads nowhere. Table~\ref{tab:tree} gives it for the three alphabets; the values are exact and
agree with an explicit tree at $N=16$ on $\Sigfour$, $N=12$ on $\Sigeight$ and on a non-symmetric case.

\begin{table*}[!t]
  \centering\small\setlength{\tabcolsep}{4pt}
  \caption{A layer-wise tree over symmetric half-stacks against the cells of
    Algorithm~\ref{alg:table}, rules and guidelines imposed; sizes summed over the depths
    $j=1,\dots,N/2$. ``Local'' keeps every prefix that obeys the manufacturing rules, ``exact''
    those with an admissible completion; ``dead'' is the share of the first not in the second;
    ``cells'' counts the cells with $F_{j}>0$.}
  \label{tab:tree}
  \begin{tabular}{l r r r r r r r}
    \toprule
    & $N$ & admissible & tree, local & tree, exact & dead & cells & tree / cells \\
    \midrule
    $\Sigfour$   & 24 & $29\,560$ & $836\,388$ & $85\,556$ & $89.8\%$ & $1\,650$ & $52$ \\
                 & 48 & $1.12\times 10^{10}$ & $3.21\times 10^{11}$ & $2.96\times 10^{10}$ & $90.8\%$ & $21\,370$ & $1.4\times 10^{6}$ \\
    $\Sigeight$  & 24 & $182\,020$ & $1.66\times 10^{8}$ & $614\,124$ & $99.63\%$ & $8\,726$ & $70$ \\
                 & 48 & $9.00\times 10^{12}$ & $1.22\times 10^{16}$ & $2.90\times 10^{13}$ & $99.76\%$ & $1\,257\,774$ & $2.3\times 10^{7}$ \\
    $\Sigtwelve$ & 24 & $782\,348$ & $2.71\times 10^{10}$ & $2.70\times 10^{6}$ & $99.990\%$ & $21\,574$ & $125$ \\
                 & 48 & $2.18\times 10^{15}$ & $3.63\times 10^{20}$ & $7.30\times 10^{15}$ & $99.998\%$ & $1.40\times 10^{7}$ & $5.2\times 10^{8}$ \\
    \bottomrule
  \end{tabular}
\end{table*}

At $N=48$ on $\Sigeight$ the dead share
is zero at depth $6$, $2\%$ at depth $9$, $14\%$ at $12$, $75\%$ at $18$ and $97\%$ at $21$ of
$24$: the dead branches sit near the mid-plane, where the tree is widest.

Three things follow.

First, counting with a tree means reaching every admissible leaf. By Lemma~\ref{lem:cells} a tree
that never enters a dead branch is the unfolding of the cells: two prefixes in the same cell root
identical subtrees, so the program keeps one copy of each. The tree grows exponentially with the
depth, the cells polynomially (Theorem~\ref{thm:poly}).

Second, a branch and bound that calls \textsc{Completable}, a lookup, at each node never enters a
dead branch, and its branching, its bound and its incumbent are untouched. On the eight- and
twelve-orientation sets the test removes more than $99\%$ of what a locally checked tree would
visit before its bound acts.

Third, with the counters uncapped, the cells reachable from a
node are the compositions its admissible completions can still attain, so a bound on a quantity that
depends on the plies only through their counts, such as the membrane stiffness, can be taken over those compositions alone.

\section{Exact uniform sampling}
\label{sec:sampling}

Algorithm~\ref{alg:sample} draws every admissible word with probability $1/|\lang|$, in $L-1$
steps with no rejection (Theorem~\ref{thm:poly}).  The sampler never fails: it either returns a compliant laminate or reports that the
language is empty, which it knows from the count.

\section{A crossover valid by construction}
\label{sec:search}

Proposition~\ref{prop:cut} gives two conditions under which a one-point cut of two compliant
half-stacks yields a compliant child. The Parikh vector of the child is a difference of prefix and
suffix counters, so both conditions, equality of states and the Parikh test, are integer tests of cost $O(1)$ per candidate cut, and a child
costs $O(m)$ once the prefix states and counters of the two parents are in hand.

\begin{proposition}
  \label{prop:cut}
  Let $h_{A}$ and $h_{B}$ be compliant half-stacks of $m$ plies and let $1\le i<m$. If the automaton
  state after $h_{A}[1..i]$ equals the state after $h_{B}[1..i]$, and if the Parikh vector of
  $h_{A}[1..i]h_{B}[i+1..m]$ satisfies the conditions of Lemma~\ref{lem:parikh-half}, then that
  word is a compliant half-stack.
\end{proposition}

\begin{proof}
  The state $(a,r)$ of Theorem~\ref{thm:regular} determines which continuations are admissible,
  since \eqref{eq:delta} depends on the prefix only through it. The prefix $h_{A}[1..i]$ is
  admissible because $h_{A}$ is, and it leaves the automaton in the state from which $h_{B}$ read
  its own suffix; that suffix is therefore readable, and the child satisfies (M1) and (M2). Its
  final state is the final state of $h_{B}$, which has run one because $h_{B}$ is compliant, so
  Lemma~\ref{lem:mirror} applies. The composition conditions hold by hypothesis.
\end{proof}

\begin{algorithm}[!htbp]
  \caption{\textsc{Cut}$(h_{A},h_{B})$: recombination inside the compliant language.}
  \label{alg:cut}
  \begin{algorithmic}[1]
    \State $q^{A}_{i},q^{B}_{i}\gets$ automaton states after $h_{A}[1..i]$, $h_{B}[1..i]$, for $1\le i<m$
    \State $\boldsymbol{c}^{A}_{i},\boldsymbol{c}^{B}_{i}\gets$ Parikh vectors of the same prefixes
    \State $I\gets\{\,i : q^{A}_{i}=q^{B}_{i}$ and $\boldsymbol{c}^{A}_{i}+\parikh{h_{B}}-\boldsymbol{c}^{B}_{i}$ satisfies Lemma~\ref{lem:parikh-half}$\,\}$
    \If{$I\neq\emptyset$} \State choose $i\in I$; \Return $h_{A}[1..i]\,h_{B}[i+1..m]$
      \Comment{compliant by Prop.~\ref{prop:cut}}
    \Else \State \Return \textsc{Sample}$(F)$ \Comment{no admissible cut: fresh exact draw}
    \EndIf
  \end{algorithmic}
\end{algorithm}

At $N=48$ an admissible cut exists for $49.0\%$ of $3000$ independent pairs on $\Sigfour$ and
$18.0\%$ on $\Sigeight$. We verified compliance on the $2008$ children with rule checkers written
separately from the dynamic program, and the sufficiency of the condition by brute force over
$3600$ pairs and all their cuts, at $N=16$ and $20$ on $\Sigfour$ and $N=12$ and $14$ on $\Sigeight$. 

\begin{proposition}
  \label{prop:absorbing}
  Let $h$ be a compliant half-stack of $m$ plies. Every cut of $h$ with itself is admissible under
  Proposition~\ref{prop:cut}, and every child it produces equals $h$.
\end{proposition}

\begin{proof}
  Take $h_{A}=h_{B}=h$ and any $1\le i<m$. The automaton state after $h[1..i]$ equals itself, so
  the state condition of Proposition~\ref{prop:cut} holds at every $i$; and the child
  $h[1..i]\,h[i+1..m]$ is $h$, whose Parikh vector satisfies the conditions of
  Lemma~\ref{lem:parikh-half} because $h$ is compliant.
\end{proof}

The fallback to the exact sampler, which is the only source of new material, fires exactly when no
cut exists, and therefore never fires once the population has converged.

The proofs of Propositions~\ref{prop:cut} and~\ref{prop:absorbing} use only that the automaton
state determines the admissible continuations and that acceptance depends on the final state and
the counters; both therefore hold for any operator that joins a prefix and a suffix meeting in one
state and closing the counters together. A diversity guard has to come from outside the operator.

\section{Discussion}
\label{sec:conclusions}

No result of this paper evaluates the response of a laminate: the counts, the sampler, the
completion test and the recombination concern the admissible set. Closed-form responses exist for
simple geometries under elementary loads, such as the buckling load and the natural frequencies of a
simply supported rectangular plate of a symmetric and balanced laminate, exact when the
bending-twisting terms $D_{16}$ and $D_{26}$ vanish. They do not exist for the geometries and load
cases of an industrial component, nor, beyond a few antisymmetric lay-ups, for laminates whose
couplings cannot be neglected, $\mathbf{B}\neq 0$ in a non-symmetric laminate and
$A_{16}, A_{26}\neq 0$ in an unbalanced one; there each evaluation is a finite-element solve, and a
search that spends evaluations on designs that cannot be built pays for them in solves.

Stacking-sequence optimisation is surveyed in \cite{ghiasi2009}. Four things this paper uses are
not ours. Treating the manufacturing
constraints exactly is \cite{zein2014}. Expressing the manufacturing rules as forbidden factors and pruning a prefix tree by them is
recited in a granted patent \cite{us10380319}. An evolutionary operator
that preserves a design guideline is not new either: \cite{irisarri2009} has crossover and mutation
preserving symmetry and balance, and what is new is preserving the whole modelled set in a single cut, each cut decided in $O(1)$. Nor is the sampling method ours. The backward walk proportional to the counts is the recursive
method \cite{flajolet1994}; an automaton accepting under a semilinear condition on counters is a
Parikh automaton \cite{klaedtke2003}; and conditioning the recursive method on an exact composition
is older still \cite{denise2000,bertoni2003}, with a Boltzmann counterpart \cite{bodini2010}. What
is new is therefore narrow: exact conditioning on thresholds and signed differences of the letter
counts, rather than on a prescribed composition, realised as one explicit product automaton of size
polynomial in $N$ that serves every composition at once and decides completability.

Four limitations bound the scope. First, the guidelines modelled are symmetry, balance and
the ten per cent rule: damage tolerance, ply drops and blending between adjacent panels
are not represented, and blending is not a function of one laminate's Parikh vector, so it would
not fit this construction without extension. The patent family cited above handles blending by
choosing each panel's sublaminate from the pre-generated library so as to couple adjacent panels
\cite{us10467350}. Second, the dynamic program tracks the sign pairs without a threshold,
$\pm 30$ and $\pm 60$ and on $\Sigtwelve$ also $\pm 15$ and $\pm 75$, through their difference alone, which suffices for balance but not for constraints that weigh those orientations
individually, such as the membrane-stiffness criterion of \cite{irisarri2009} on $\Sigeight$. Third, this paper measures no search: what the operations save a heuristic in evaluations and in
generation time is not quantified here. Fourth, the program is polynomial for a fixed alphabet but its degree grows with the number of sign pairs: the
non-symmetric balanced count on $\Sigtwelve$, which carries five differences, needs
$1.1\times 10^{9}$ cells at the middle level, $69$~GB and an hour and a half at $N=54$.

\section{Conclusions}
\label{sec:final}

The manufacturing rules and guidelines treated here form a regular language conditioned on
thresholds and signed differences of its letter counts, so one table per ply count, polynomial in
$N$, counts the admissible sequences exactly, draws them uniformly and tells whether a partial
sequence can still be completed. In a branch and bound it supplies the completion test without which
$99.8\%$ of a locally checked tree is dead at $48$ plies on eight orientations, and in a genetic
algorithm a recombination that returns admissible children by construction. Blending across panels,
ply drops and damage tolerance are the constraints this construction does not yet carry.

\section*{Funding}
This research did not receive any specific grant from funding agencies in the public, commercial,
or not-for-profit sectors.

\section*{Declaration of competing interest}
The authors declare that they have no known competing financial interests or personal relationships
that could have appeared to influence the work reported in this paper.

\section*{Data availability}
The program of Algorithms~\ref{alg:dfa}--\ref{alg:cut} and the meet-in-the-middle counter are
available on Zenodo~\cite{langiu2026stackcode}. The rule checkers of
Section~\ref{sec:verification}, written separately from the dynamic program, are those of the
open-source design workflow of \cite{langiu2026code}, also on Zenodo. The scripts were written and
run with the assistance of Claude (Anthropic), and every number reported was recomputed by a
separate verification script.

\section*{Declaration of generative AI and AI-assisted technologies in the manuscript preparation process}
During the preparation of this work the authors used Anthropic's Claude, through the Claude Code
interface, in order to draft and edit manuscript text, to write and run the counting, sampling and
search scripts, and to check the cited sources against the statements that cite them. After using
this tool, the authors reviewed and edited the content as needed and take full responsibility for
the content of the published article.

\bibliographystyle{plainurl}
\bibliography{paper2}

\end{document}